\documentclass[11pt]{article}

\usepackage[utf8]{inputenc}
\usepackage[T1]{fontenc}
\usepackage{mathptmx}
\usepackage[margin=1in]{geometry}
\usepackage{amsmath,amssymb,amsthm}
\usepackage{graphicx}
\usepackage{booktabs}
\usepackage{authblk}
\usepackage{xcolor}
\usepackage{caption}
\usepackage{subcaption}
\usepackage[numbers,sort&compress]{natbib}
\usepackage{algorithm}
\usepackage{algpseudocode}
\usepackage{titlesec}
\usepackage[expansion=false]{microtype}
\usepackage{enumitem}
\usepackage{float}
\usepackage[hidelinks]{hyperref}

\titleformat{\section}{\normalfont\large\bfseries}{\thesection}{0.6em}{}
\titleformat{\subsection}{\normalfont\normalsize\bfseries}{\thesubsection}{0.6em}{}

\graphicspath{{Figures/}}

\newtheorem{proposition}{Proposition}

\newtheorem{definition}{Definition}

\newcommand{\bx}{\mathbf{x}}
\newcommand{\bff}{\mathbf{f}}
\newcommand{\bz}{\mathbf{z}}
\newcommand{\bW}{\mathbf{W}}
\newcommand{\bb}{\mathbf{b}}
\newcommand{\bd}{\mathbf{d}}
\newcommand{\R}{\mathbb{R}}
\newcommand{\E}{\mathbb{E}}
\newcommand{\KL}{\mathrm{KL}}
\newcommand{\TopK}{\mathrm{TopK}}
\newcommand{\AUC}{\mathrm{AUC}}

\title{\bfseries Causal dictionary learning reveals and validates transcription-factor binding features in genomic language models}

\author[1]{Sarwan Ali}
\affil[1]{Columbia University Irving Medical Center, New York, USA

sa4559@cumc.columbia.edu}

\date{}

\begin{document}
\maketitle

\begin{abstract}
\noindent
Genomic language models achieve strong performance across regulatory-genomics
tasks, yet what these models internally represent remains opaque, and the field
lacks a principled procedure for verifying that an apparent ``concept'' inside a model is real rather than an artifact of sequence composition. We introduce a
framework that combines sparse dictionary learning with causal intervention to
extract, validate, and causally test interpretable features in genomic
foundation models. Training top-$k$ sparse autoencoders on the hidden activations
of two architecturally distinct models, Nucleotide Transformer ($6$-mer
tokenization) and DNABERT-2 (byte-pair encoding), we recover thousands of
monosemantic features that map to transcription-factor (TF) sequence motifs. We
show that the naive validation of such features against position weight matrices
is severely confounded by GC composition and repetitive elements, producing
hundreds of spurious ``TF features'', and we develop a composition-matched,
binding-resolved protocol that removes these confounds. Critically, we move
beyond correlation: by ablating individual dictionary directions during the
model's forward pass and measuring the induced shift in the model's own
predictive distribution, we establish that specific features are
\emph{causally} used to represent cell-type-specific TF binding, not merely motif
presence. Across three transcription factors (CTCF, GATA1, REST) and both
architectures, causally validated binding features emerge reproducibly
($7$--$14$ of $15$ tested features per condition), while two classes of negative
control, scrambled binding labels and randomly selected features, yield no
detectable signal. The framework is purely computational, uses only public data,
and provides a reusable standard for interpretability claims in genomic deep
learning.
\end{abstract}

\section{Introduction}

DNA language models trained by self-supervision on reference genomes have rapidly
become general-purpose tools in regulatory genomics, supporting variant-effect
prediction, chromatin-state annotation and regulatory-element
discovery~\citep{dallatorre2025nucleotide,zhou2024dnabert2,schiff2024caduceus,
nguyen2023hyenadna,nguyen2024evo,benegas2023gpn}. Their predictive utility is now
well documented, but their interpretability is not. Unlike supervised models
built around designed features~\citep{zhou2015deepsea,kelley2018basenji,
avsec2021enformer}, foundation models distribute information across
high-dimensional activations in which individual neurons are
polysemantic, responding to many unrelated sequence properties at
once~\citep{elhage2022superposition}. Consequently, it is rarely possible to
state what a genomic language model has actually learned, which limits both
scientific trust and the use of these models for hypothesis generation about
genome regulation~\citep{novakovsky2023obtaining}.

In natural-language processing, sparse autoencoders (SAEs) have emerged as a
powerful remedy for polysemanticity: by reconstructing model activations through
an over-complete, sparsely active dictionary, they decompose distributed
representations into approximately monosemantic
features~\citep{cunningham2023sparse,bricken2023monosemanticity,
templeton2024scaling,gao2024scaling}. Whether this paradigm transfers to genomic
models, and whether the resulting features correspond to bona fide biological
entities, has not been established. Two obstacles stand in the way. First,
genomic models differ fundamentally in how they tokenize DNA, fixed $k$-mers,
learned byte-pair vocabularies, or single nucleotides, and it is unknown whether
dictionary learning behaves consistently across these schemes. Second, and more
seriously, the genome's sequence statistics make naive validation treacherous:
many transcription-factor motifs are GC-rich, as are CpG islands and the most
abundant repetitive elements in the human
genome~\citep{deininger2011alu,schmitges2016ctcfalu}, so a feature that merely
detects local GC content or an \emph{Alu} repeat will spuriously appear
``enriched'' for any GC-rich motif. An interpretability method that does not
control for this will systematically over-report success.

We address both obstacles and, in doing so, propose a general standard. Our
contributions are as follows. (i)~We train top-$k$ SAEs on the activations of two
genomic language models with different tokenizations and show that monosemantic,
motif-aligned features arise robustly in both, with middle layers yielding the
most interpretable dictionaries. (ii)~We demonstrate quantitatively that
position-weight-matrix enrichment, the obvious validation, is dominated by
compositional confounds, and we introduce a GC-matched, binding-resolved test
that isolates genuine signal. (iii)~Most importantly, we develop a
\emph{causal} validation: ablating a single dictionary direction during the
forward pass and measuring the Kullback--Leibler shift in the model's masked-token
predictions at bound versus unbound motif sites. This distinguishes features the
model \emph{uses} to represent binding from features that merely correlate with
it. (iv)~We validate the whole framework with stringent negative controls and
show cross-architecture, cross-TF reproducibility. The result is a purely
computational, fully reproducible procedure, built entirely on public models and
data, for making, and disciplining, interpretability claims about genomic deep
learning.

\section{Related Work}

\paragraph{Genomic language models.}
Self-supervised models of DNA have progressed rapidly from early $k$-mer
transformers~\citep{ji2021dnabert} to large multi-species encoders such as the
Nucleotide Transformer~\citep{dallatorre2025nucleotide} and the
byte-pair-encoded DNABERT-2~\citep{zhou2024dnabert2}, and to long-context
architectures including HyenaDNA~\citep{nguyen2023hyenadna}, the state-space model
Caduceus~\citep{schiff2024caduceus} and the genome-scale Evo
model~\citep{nguyen2024evo,gu2023mamba}. These models support competitive
variant-effect prediction~\citep{benegas2023gpn} and regulatory annotation, but
recent evaluations question how much regulatory signal their representations
actually contain~\citep{tang2024evaluating}, underscoring the need for tools that
probe what they encode rather than only how well they predict. Supervised
sequence models such as DeepSEA~\citep{zhou2015deepsea},
Basenji~\citep{kelley2018basenji} and Enformer~\citep{avsec2021enformer} remain
the performance reference for binding and expression prediction; our aim is not
to compete with them but to interpret the self-supervised models.

\paragraph{Interpretability of deep models in genomics.}
Attribution methods such as DeepLIFT~\citep{shrikumar2017deeplift} and
SHAP~\citep{lundberg2017shap}, together with in-silico mutagenesis, are the
standard tools for explaining genomic neural
networks~\citep{novakovsky2023obtaining}. These produce per-input importance maps
but do not yield a global vocabulary of reusable features, and they inherit the
polysemanticity of individual neurons~\citep{elhage2022superposition}. Our work
instead learns an explicit, model-wide dictionary and asks which of its elements
the model causally uses.

\paragraph{Sparse autoencoders and mechanistic interpretability.}
Sparse dictionary learning on model activations has become a leading approach to
extracting monosemantic features from language
models~\citep{cunningham2023sparse,bricken2023monosemanticity,templeton2024scaling},
with the top-$k$ formulation offering a clean sparsity--fidelity
trade-off~\citep{gao2024scaling}. Validating such features is an open problem:
activation patching and causal mediation establish that components are used by a
model~\citep{vig2020causal,meng2022locating}, but apparent feature interpretations
can be illusory if not causally grounded~\citep{makelov2023subspace}. We bring
this causal-intervention discipline to genomics, where, uniquely, an external
ground truth, experimental ChIP-seq binding, is available, and where
sequence-composition confounds make naive validation especially hazardous.

\paragraph{Transcription-factor binding resources.}
We ground features against curated motif models from
JASPAR~\citep{rauluseviciute2024jaspar} and HOCOMOCO~\citep{vorontsov2024hocomoco}
and against experimental occupancy from ENCODE
ChIP-seq~\citep{encode2012,moore2020encode3}. The biology of our test factors, 
CTCF as a GC-rich insulator~\citep{ong2014ctcf} whose motifs are frequently
embedded in \emph{Alu} repeats~\citep{deininger2011alu,schmitges2016ctcfalu}, is
precisely what makes composition control essential.

\section{Methodology}

We formalize the three components of the framework, the sparse dictionary, the
composition-matched binding test, and the causal ablation, and then state two
propositions that explain why the design isolates binding-specific causal use.
The full procedure is summarized in Algorithm~\ref{alg:framework}.

\subsection{Top-$k$ sparse autoencoders}

Let $\bx\in\R^{d}$ denote a (mean-centered, scaled) model activation. A top-$k$
sparse autoencoder~\citep{gao2024scaling} encodes $\bx$ into a sparse code
$\bff\in\R^{m}$ with $m=16d$ and decodes it back:
\begin{align}
\bz &= \bW_{\mathrm{enc}}(\bx-\bb_{\mathrm{pre}}), \label{eq:enc}\\
\bff &= \TopK_k\!\big(\mathrm{ReLU}(\bz)\big), \label{eq:topk}\\
\hat{\bx} &= \bW_{\mathrm{dec}}\bff + \bb_{\mathrm{pre}}, \label{eq:dec}
\end{align}
where $\TopK_k$ retains the $k$ largest entries and zeroes the rest, the columns
$\bd_j$ of $\bW_{\mathrm{dec}}\in\R^{d\times m}$ are constrained to unit norm, and
$\bb_{\mathrm{pre}}\in\R^{d}$ is a learned pre-bias. The training objective is
reconstruction error with an auxiliary term that revives inactive features,
\begin{equation}
\mathcal{L}
= \underbrace{\big\|\bx-\hat{\bx}\big\|_2^2}_{\text{reconstruction}}
+ \alpha\,\underbrace{\big\|(\bx-\hat{\bx}) - \hat{\bx}_{\mathrm{aux}}\big\|_2^2}_{\text{dead-feature revival}},
\label{eq:loss}
\end{equation}
where $\hat{\bx}_{\mathrm{aux}}$ reconstructs the residual using only the top
dead features and $\alpha=1/16$. The fixed sparsity $\|\bff\|_0=k$ is enforced
architecturally by $\TopK_k$, avoiding the shrinkage bias of $\ell_1$ penalties.
We initialized dictionary directions from random data activations, which proved
essential for keeping the dictionary alive (raising the live fraction from
$45\%$ to $93\%$ at fixed reconstruction error).

\subsection{Composition-matched binding test}

For each transcription factor we build three classes of length-matched genomic
window: \emph{bound} (ChIP-seq peak centers containing a strong motif),
\emph{unbound-motif} (strong motif occurrences outside all peaks), and
\emph{background} (random windows). A position weight matrix $M\in\R^{L\times4}$
from JASPAR~\citep{rauluseviciute2024jaspar} is converted to a log-odds matrix
against a uniform background, and a window is deemed to contain a strong motif if
its maximum log-odds score over all offsets exceeds a fixed fraction of the
matrix maximum. To remove the compositional confound, the unbound-motif set is
GC-matched to the bound set by rejection sampling against the bound GC histogram.

For a feature $j$ we max-pool its activation $f_j$ over the tokens of each window
and test, with a one-sided Mann--Whitney $U$~\citep{mann1947test}, whether
$f_j^{\text{bound}}>f_j^{\text{unbound}}$. We summarize effect size by
$\AUC_{j}=U/(n_{\mathrm{b}}n_{\mathrm{u}})$, the probability that a random bound
window activates the feature more strongly than a random unbound-motif window. A
feature is \emph{motif-selective} if $f_j^{\text{unbound}}>f_j^{\text{background}}$
and \emph{binding-sensitive} if $\AUC_j>0.55$ at Bonferroni-corrected
significance; we additionally require $|\mathrm{corr}(f_j,\mathrm{GC})|<0.2$ for
``GC-robust'' status.

\subsection{Causal feature ablation}

Let $h_\ell(\bx)$ be the hidden state at the SAE's layer and
$g_{\ell\to L}(\cdot)$ the remainder of the network mapping that hidden state to
output logits over the token vocabulary. For dictionary direction $\bd_j$ with
encoded activation $f_j$, the ablated hidden state removes the feature's
contribution in the (normalized) SAE space,
\begin{equation}
h_\ell'(\bx) = h_\ell(\bx) - s^{-1}\,\big(f_j(\bx)\,\bd_j\big),
\label{eq:ablate}
\end{equation}
where $s$ is the activation scale used during SAE training. We then complete the
forward pass with the perturbed state and compare the model's predictive
distributions at each token position $t$ via the Kullback--Leibler divergence,
\begin{equation}
\Delta_j(\bx)
= \frac{1}{T}\sum_{t=1}^{T}
\KL\!\Big(p_t\big(g_{\ell\to L}(h_\ell)\big)\,\big\|\,
            p_t\big(g_{\ell\to L}(h_\ell')\big)\Big),
\label{eq:kl}
\end{equation}
where $p_t(\cdot)=\mathrm{softmax}$ of the logits at position $t$ and $T$ is the
number of tokens. The \emph{binding-specific causal effect} of feature $j$ is the
one-sided Mann--Whitney comparison of $\{\Delta_j(\bx)\}$ between bound and
unbound-motif windows, summarized again by an $\AUC$. A feature passes if this
$\AUC>0.55$ at Bonferroni-corrected significance.

\begin{algorithm}[H]
\caption{Causal dictionary learning and validation}
\label{alg:framework}
\begin{algorithmic}[1]
\State \textbf{Input:} genomic model with layer $\ell$; windows; ChIP-seq peaks; motif PWMs
\State \textbf{Extract} per-token activations $\{\bx\}$ at layer $\ell$ via forward hooks
\State \textbf{Train} top-$k$ SAE on $\{\bx\}$ by minimizing Eq.~\eqref{eq:loss} \Comment{Eqs.~\eqref{eq:enc}--\eqref{eq:dec}}
\State \textbf{Build} bound / unbound-motif (GC-matched) / background window sets
\For{each live feature $j$}
  \State compute $\AUC_j$ (bound vs.\ unbound) and $\mathrm{corr}(f_j,\mathrm{GC})$
  \State mark $j$ \emph{binding-sensitive} if $\AUC_j>0.55$ (Bonferroni) and $|\mathrm{corr}|<0.2$
\EndFor
\For{each binding-sensitive feature $j$ (and random controls)}
  \State ablate $\bd_j$ (Eq.~\eqref{eq:ablate}); measure $\Delta_j$ (Eq.~\eqref{eq:kl}) at bound vs.\ unbound
  \State mark $j$ \emph{causally validated} if binding-specific $\AUC>0.55$ (Bonferroni)
\EndFor
\State \textbf{Output:} causally validated binding features per factor and model
\end{algorithmic}
\end{algorithm}

\subsection{Theoretical analysis}

We formalize why the causal ablation test isolates binding-specific use, and why
the composition-matched design is necessary. Throughout, let $B\in\{0,1\}$
indicate ChIP-seq binding, let $S$ denote the observable sequence of a window,
and let $f_j(S)$ be feature $j$'s (max-pooled) activation.

\begin{proposition}[Spurious enrichment from composition]
\label{prop:confound}
Let $\rho(S)\in[0,1]$ be a compositional summary of $S$ (e.g.\ GC fraction), and
suppose a motif log-odds score $\sigma(S)$ satisfies
$\E[\sigma(S)\mid \rho(S)=r]$ strictly increasing in $r$. If a feature is a pure
composition detector, $f_j(S)=\phi(\rho(S))$ for some increasing $\phi$, then in
an unmatched comparison where the high-activation set has higher mean composition
than background, the feature's top windows are enriched for high $\sigma$, i.e.\
\[
\E\big[\sigma(S)\,\big|\, f_j(S)\ \text{large}\big]
> \E\big[\sigma(S)\big],
\]
even though $f_j \perp B \mid \rho$, i.e.\ the feature carries no information
about binding beyond composition.
\end{proposition}

\begin{proof}
Since $f_j=\phi(\rho)$ with $\phi$ increasing, the event $\{f_j \text{ large}\}$
equals $\{\rho > r_0\}$ for a threshold $r_0$ above the marginal median of
$\rho$. By the monotonicity of $\E[\sigma\mid\rho=r]$ in $r$ and the law of total
expectation,
\[
\E[\sigma\mid \rho>r_0]
=\!\int_{r_0}^{1}\!\E[\sigma\mid\rho=r]\,dF_{\rho\mid \rho>r_0}(r)
> \!\int_0^1 \!\E[\sigma\mid\rho=r]\,dF_\rho(r)
=\E[\sigma],
\]
because conditioning on $\rho>r_0$ stochastically increases $\rho$ and
$\E[\sigma\mid\rho=r]$ is increasing. The independence $f_j\perp B\mid\rho$ holds
by construction since $f_j$ is a deterministic function of $\rho$ alone. Hence
the feature shows motif enrichment while being conditionally uninformative about
binding.
\end{proof}

Proposition~\ref{prop:confound} formalizes the empirical failure of naive
enrichment and motivates comparing bound against \emph{GC-matched} unbound
windows: matching equalizes the marginal law of $\rho$ across groups, so a pure
composition detector satisfies $\E[f_j\mid B{=}1]=\E[f_j\mid B{=}0]$ in the
matched population and its binding-sensitivity $\AUC\to 1/2$. The
composition-matched test is therefore \emph{calibrated} against this confound.

\begin{definition}[Causal binding use]
Feature $j$ is \emph{causally used for binding} if its ablation changes the
model's output distribution more on bound than on unbound-motif inputs, i.e.\
$\E[\Delta_j \mid B=1] > \E[\Delta_j \mid B=0]$ on motif-matched inputs.
\end{definition}

\begin{proposition}[Separation of importance from binding-specificity]
\label{prop:sep}
Let the contribution removed by ablation be $c_j(S)=s^{-1}f_j(S)\bd_j$ and write
the post-layer map's local sensitivity as $J(S)=\partial g_{\ell\to L}/\partial
h_\ell$. To first order,
\begin{equation}
\Delta_j(S)\ \approx\ \tfrac{1}{2}\,c_j(S)^\top G(S)\,c_j(S),
\label{eq:taylor}
\end{equation}
where $G(S)=J(S)^\top H_{\KL}\,J(S)\succeq 0$ and $H_{\KL}$ is the Fisher
information of the softmax output. Then a feature with large average effect
$\E[\Delta_j]$ but $f_j\perp B$ (motif-only) has
$\E[\Delta_j\mid B=1]=\E[\Delta_j\mid B=0]$ on motif-matched inputs, hence
binding-specificity $\AUC=1/2$; whereas a feature whose activation is elevated on
bound inputs, $\E[f_j\mid B=1]>\E[f_j\mid B=0]$, yields
$\E[\Delta_j\mid B=1]>\E[\Delta_j\mid B=0]$ whenever $G(S)$ does not
anti-correlate with binding.
\end{proposition}

\begin{proof}
Equation~\eqref{eq:taylor} is the second-order Taylor expansion of the KL
divergence between $p_t(g(h_\ell))$ and $p_t(g(h_\ell-c_j))$ about $c_j=0$; the
first-order term vanishes because the KL divergence and its gradient are zero at
$c_j=0$, and the quadratic form uses the Fisher information $H_{\KL}$ as the
Hessian of the KL divergence at the unperturbed distribution, propagated through
$J(S)$ by the chain rule. Taking $c_j(S)=s^{-1}f_j(S)\bd_j$ gives
$\Delta_j(S)\approx \tfrac12 s^{-2} f_j(S)^2\,\bd_j^\top G(S)\bd_j$. Write
$\gamma(S)=\tfrac12 s^{-2}\bd_j^\top G(S)\bd_j \ge 0$. For a motif-only feature,
$f_j\perp B$ on motif-matched inputs, and if the geometric factor $\gamma$ does not anti-correlate with binding (in particular if it is independent of $B$, true to first order when ablation does not systematically interact with binding), then $\E[f_j^2\gamma\mid B=1]=
\E[f_j^2\gamma\mid B=0]$, giving equal conditional effects and $\AUC=1/2$. For a
binding-sensitive feature, $\E[f_j^2\mid B=1]>\E[f_j^2\mid B=0]$ since $f_j\ge0$
and its mean is larger on bound inputs; multiplying by the nonnegative,
binding-independent $\gamma$ preserves the inequality, so
$\E[\Delta_j\mid B=1]>\E[\Delta_j\mid B=0]$ and the binding-specificity $\AUC>1/2$.
\end{proof}

Proposition~\ref{prop:sep} predicts a counterintuitive phenomenon we observe
empirically (Section~\ref{sec:results}): a motif-only feature can have a far
larger \emph{absolute} ablation effect than a binding feature, yet a
binding-specificity $\AUC$ at chance, because absolute effect is governed by
$\E[f_j^2\gamma]$ while specificity is governed by the \emph{difference} of this
quantity across binding states. The causal test reads the latter.

\section{Experimental Setup}

\paragraph{Models.}
We used the publicly released checkpoints of Nucleotide Transformer (NT) v2
($500$M parameters, $6$-mer tokenization, $24$ transformer
blocks)~\citep{dallatorre2025nucleotide} and DNABERT-2 ($117$M parameters,
byte-pair encoding, $12$ blocks)~\citep{zhou2024dnabert2}. Both expose a
masked-language-model head, enabling the logit-level causal readout of
Eq.~\eqref{eq:kl}.

\paragraph{Data.}
The human reference genome (hg38), ENCODE SCREEN candidate cis-regulatory
elements~\citep{moore2020encode3}, JASPAR~2024 CORE vertebrate
motifs~\citep{rauluseviciute2024jaspar}, and ENCODE ChIP-seq IDR peak sets for
CTCF (GM12878), GATA1 (K562) and REST (K562) were obtained from their public
repositories. For activation extraction we sampled $5\times10^4$ cis-regulatory
and $5\times10^4$ random $200$\,bp windows, retaining only windows with
unambiguous nucleotides, yielding $3.6$--$4.2\times10^{6}$ token activations per
layer. Because the two models tokenize DNA differently, NT into non-overlapping
$6$-mers, DNABERT-2 into variable-length byte-pair tokens, we recorded, for every
token, the exact base-pair span it covers, mapping all features to genomic
coordinates regardless of tokenization. Special tokens were excluded from all
analyses.

\paragraph{Training and layer selection.}
For each model we sampled four evenly spaced layers and trained one top-$k$ SAE
per layer (dictionary size $16\times$ the hidden width; $k=32$) for one epoch with
Adam (learning rate $4\times10^{-4}$, batch size $4096$). Layer selection used
reconstruction fidelity and live-feature fraction (Section~\ref{sec:results},
Fig.~\ref{fig:layers}).

\paragraph{Binding sets and controls.}
Per-factor window sets were built as in Methodology, with the unbound-motif set
GC-matched to bound. Two negative controls were processed by the identical
pipeline: \textsc{scramble} (random windows with random bound/unbound labels) and
\textsc{gata1-scram} (real GATA1 motif windows with randomized labels, so that
sequence and motif content are unchanged and only the binding label is
destroyed). Population causal analyses ablate the top $15$ binding-sensitive,
GC-robust features and $15$ randomly selected live features per condition.

\paragraph{Statistics.}
All $P$-values use the one-sided Mann--Whitney $U$ test~\citep{mann1947test} and
are Bonferroni-corrected for the number of features tested; significance
thresholds are stated with each analysis.

\section{Results}
\label{sec:results}

\subsection{Sparse dictionaries are monosemantic and layer-dependent}

The autoencoders reconstructed activations faithfully while remaining sparse:
the best dictionaries explained $74$--$76\%$ of activation variance with only
$32$ of $\sim$12{,}000--16{,}000 features active per token. Reconstruction quality
and the fraction of ``live'' (ever-active) features both peaked at intermediate
depth (Fig.~\ref{fig:layers}), mirroring observations in language
models~\citep{bricken2023monosemanticity}: NT layer~$14$ ($50\%$ depth) and
DNABERT-2 layer~$6$ ($50\%$ depth) provided the best balance of fidelity and
dictionary utilization and were used in all subsequent analyses. Both
tokenization schemes therefore admit high-quality sparse dictionaries, despite
their very different vocabularies.

\begin{figure}[h!]
  \centering
  \includegraphics[width=0.6\textwidth]{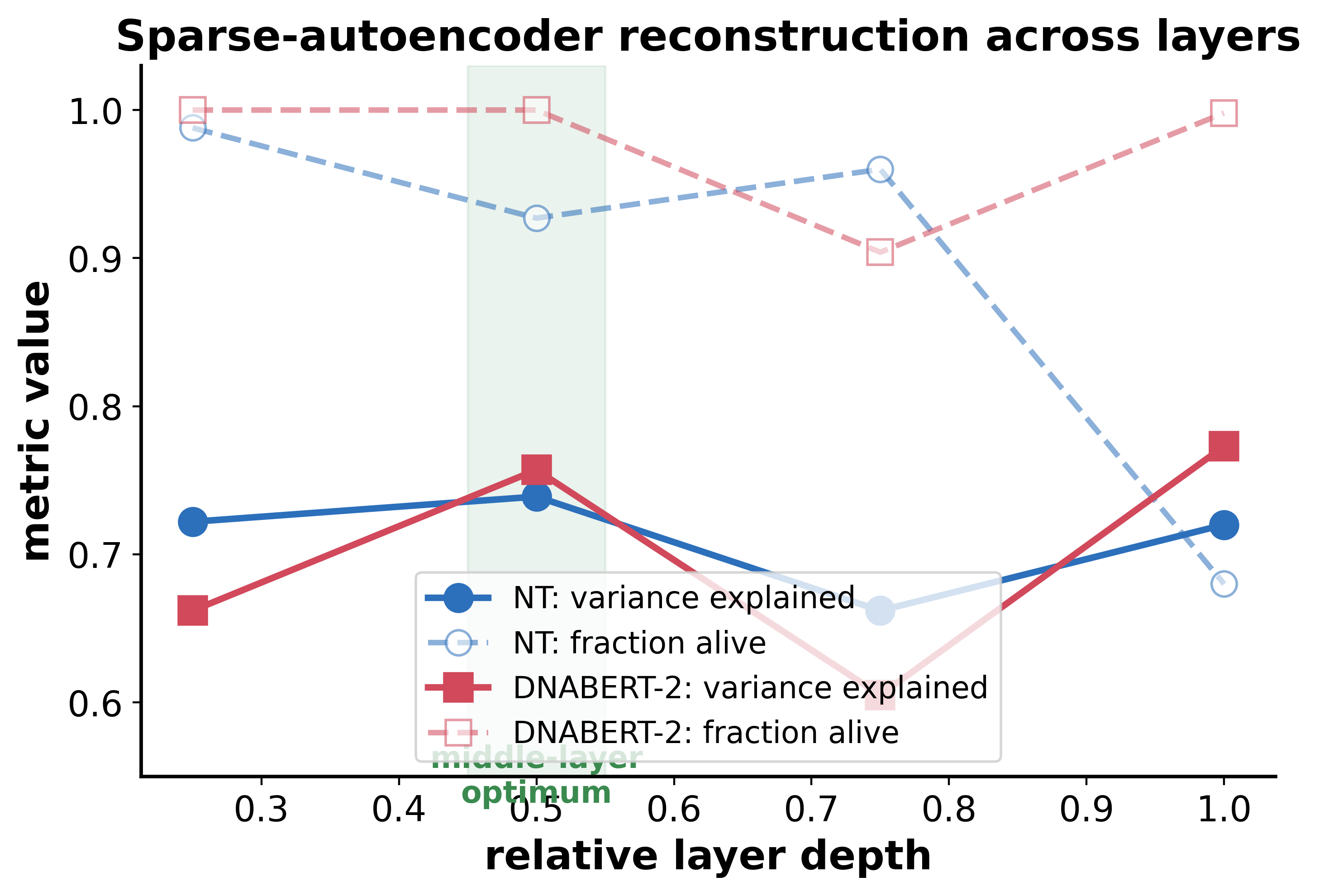}
  \caption{\textbf{Sparse-autoencoder reconstruction across model depth.}
  Variance explained (solid) and fraction of live dictionary features (dashed)
  for Nucleotide Transformer (blue) and DNABERT-2 (red) as a function of relative
  layer depth. Both architectures attain their best trade-off of reconstruction
  fidelity and feature utilization at intermediate depth (shaded), motivating the
  use of NT layer~$14$ and DNABERT-2 layer~$6$.}
  \label{fig:layers}
\end{figure}

\subsection{Naive motif enrichment is confounded by sequence composition}

To ask whether dictionary features correspond to regulatory grammar, we first
applied the obvious test: for each feature, we collected the genomic windows on
which it fired most strongly and scored them against transcription-factor
position weight matrices~\citep{rauluseviciute2024jaspar}. Using CTCF, a
ubiquitous insulator with a long, information-rich motif~\citep{ong2014ctcf},
this naive procedure flagged $440$ of $6{,}576$ tested features as significantly
CTCF-enriched.

Inspection revealed this to be largely an artifact. The single most
``CTCF-enriched'' feature fired almost exclusively on near-identical copies of the
\emph{Alu} consensus, the most abundant repetitive element in the human
genome~\citep{deininger2011alu}, which scores above background on the GC-rich
CTCF matrix without representing CTCF binding; other top features detected generic
GC-rich or CpG-island sequence. Because CTCF, CpG islands and \emph{Alu} elements
are all GC-rich, a feature that merely tracks base composition is spuriously
``enriched'' for CTCF. Comparing each feature's motif scores against a
\emph{GC-matched} background rather than an unmatched one collapsed the apparent
enrichment of these features toward zero, identifying them as compositional, as
predicted by Proposition~\ref{prop:confound}. The count of motif-selective
features for each factor after composition control is summarized in
Fig.~\ref{fig:motifcounts}.

\begin{figure}[h!]
  \centering
  \includegraphics[width=0.56\textwidth]{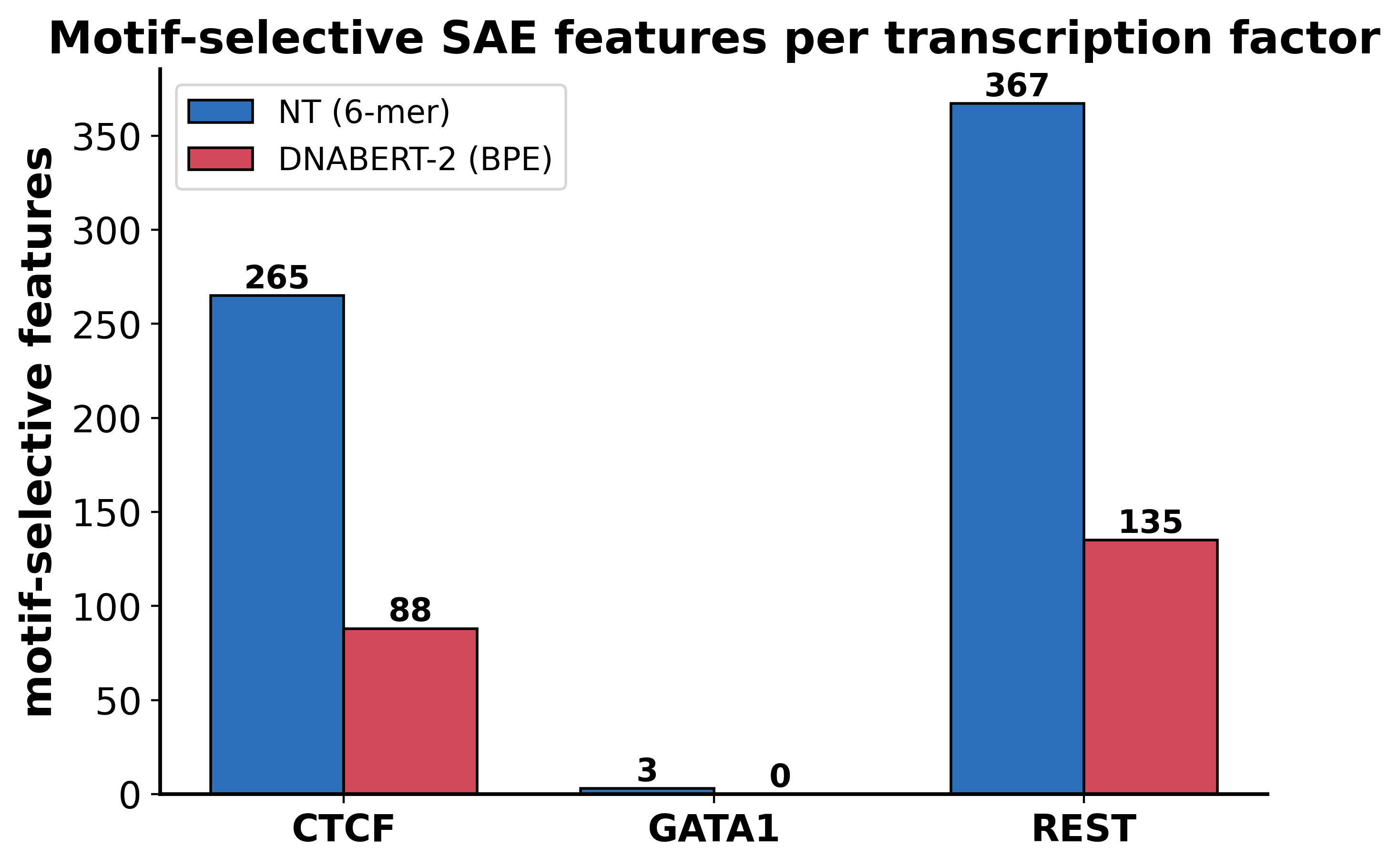}
  \caption{\textbf{Motif-selective features per transcription factor} after
  composition control, for NT (blue) and DNABERT-2 (red). CTCF and REST, long,
  information-rich motifs, yield many motif-selective features; GATA1, with a
  short \texttt{WGATAR} motif, yields almost none ($3$ and $0$, respectively), foreshadowing that its binding is
  encoded through context rather than the canonical motif.}
  \label{fig:motifcounts}
\end{figure}

\subsection{A composition-matched, binding-resolved test isolates genuine signal}

The deeper question is whether a feature encodes \emph{binding}, cell-type-specific
occupancy measured by ChIP-seq, rather than mere motif presence. Testing whether
each feature's activation is higher on bound than on GC-matched unbound-motif
windows, NT layer~$14$ contained $265$ motif-selective CTCF features but only a
small number that remained binding-sensitive after GC control; the cleanest,
feature~$8087$, exhibited a monotonic activation gradient from background to
unbound-motif to bound windows (Fig.~\ref{fig:exemplar}a) with negligible GC
correlation. A stringent, composition-aware test therefore recovers a specific,
interpretable, binding-associated feature where the naive test had returned
hundreds of confounded ones.

\subsection{Feature ablation establishes causal use of binding representations}

Correlation between activation and binding does not establish that the model
\emph{uses} a feature to represent binding. We therefore ablated single
dictionary directions during the forward pass (Eq.~\eqref{eq:ablate}) and measured
the Kullback--Leibler shift in masked-token predictions (Eq.~\eqref{eq:kl}). For
the exemplar CTCF feature, ablation shifted NT's predictions significantly more at
bound than at unbound sites ($\AUC=0.63$; $P=1.7\times10^{-18}$). Two controls
confirmed specificity (Fig.~\ref{fig:exemplar}b). Ablating a \emph{random} feature
produced no binding-specific effect ($\AUC=0.51$). More tellingly, ablating a
\emph{motif-selective but not binding-sensitive} feature produced a large overall
prediction shift, an order of magnitude larger in absolute terms than the binding
feature, yet \emph{no} bound-versus-unbound differential ($\AUC=0.51$), exactly
the dissociation predicted by Proposition~\ref{prop:sep}. The causal test thus
distinguishes features that matter to the model in general from features that
matter \emph{specifically for binding}, which a magnitude-based analysis cannot.

\begin{figure}[h!]
  \centering
  \includegraphics[width=0.9\textwidth]{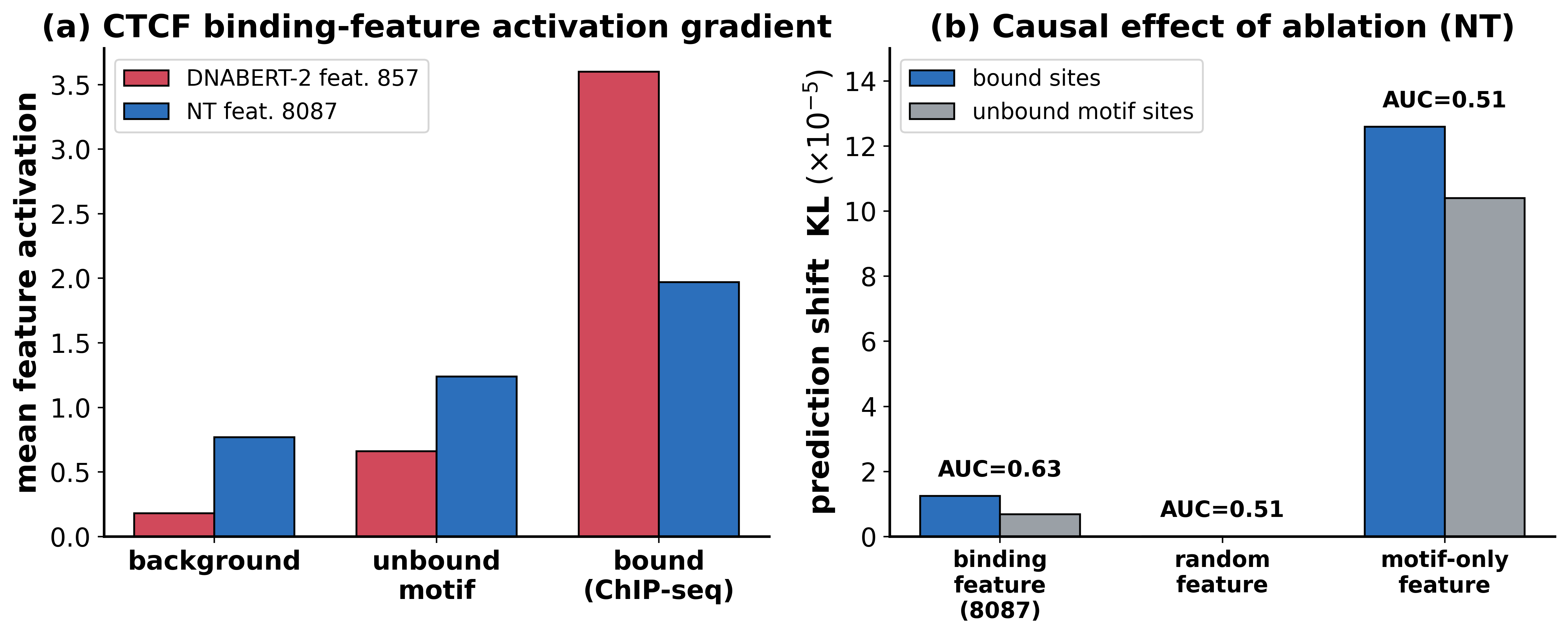}
  \caption{\textbf{A causally validated CTCF binding feature.}
  \textbf{(a)}~Mean activation of two exemplar binding features (DNABERT-2
  feature~$857$, red; NT feature~$8087$, blue) across background, unbound-motif
  and ChIP-seq-bound windows; activation increases monotonically with binding,
  not merely with motif presence. \textbf{(b)}~Causal effect of ablating each of
  three NT features, measured as the prediction-distribution shift
  (Kullback--Leibler divergence) at bound (blue) versus unbound-motif (grey)
  sites. The binding feature shows a binding-specific effect ($\AUC=0.63$); a
  random feature shows none; a motif-only feature shows a large overall effect
  but no binding specificity ($\AUC=0.51$), demonstrating that the test isolates
  binding-specific causal use rather than general feature importance.}
  \label{fig:exemplar}
\end{figure}

\subsection{Causal binding features are reproducible across factors and
architectures}

We next asked whether causal binding features are a population phenomenon and
whether the framework generalizes. For each transcription factor and model we
ablated the top $15$ binding-sensitive, GC-robust features and $15$ random
features. Across three factors of distinct structural class, CTCF (insulator),
GATA1 (lineage-specific activator) and REST (repressor), and both architectures,
a substantial fraction of binding features showed significant binding-specific
causal effects, whereas random-feature controls showed essentially none
(Fig.~\ref{fig:matrix}, Fig.~\ref{fig:auc}, Table~\ref{tab:results}). The effect
was, if anything, stronger in DNABERT-2 than in NT, indicating that byte-pair
tokenization does not impede, and may sharpen, the emergence of binding
representations.

A biologically informative subtlety emerged for GATA1. Its binding features were
only weakly motif-selective in the correlational test
(Fig.~\ref{fig:motifcounts}), which we had initially read as weak representation.
The causal test overturned this: GATA1 binding features were robustly causal in
both models ($9/15$ and $13/15$; Table~\ref{tab:results}), indicating that the
models encode GATA1 occupancy through contextual sequence features beyond the
canonical \texttt{WGATAR} motif, signal that the causal test detects but
motif-based correlation misses. This illustrates the added value of intervention
over association.

\begin{figure}[h!]
  \centering
  \includegraphics[width=0.58\textwidth]{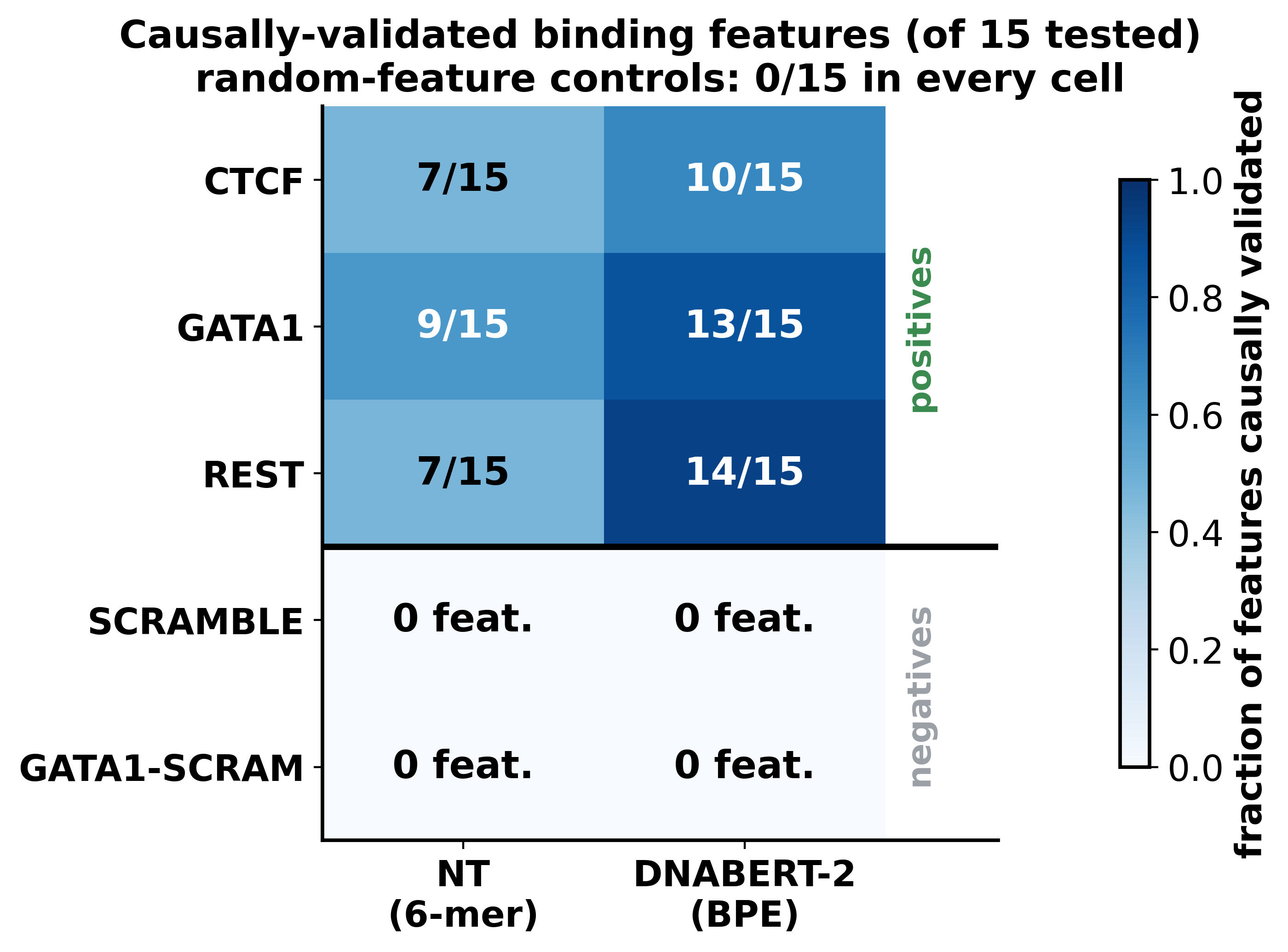}
  \caption{\textbf{Causal validation matrix.} Number of features (of $15$ tested)
  with a significant binding-specific causal effect (Bonferroni-corrected
  $P<1.7\times10^{-3}$ and $\AUC>0.55$), for each transcription factor (rows) and
  model (columns). Random-feature controls yielded $0/15$ in every cell. The two
  bottom rows are negative controls in which binding labels were scrambled; these
  produced no binding-sensitive features and hence nothing to test (``$0$
  feat.''), confirming that the framework reports a true null when no real
  binding signal is present.}
  \label{fig:matrix}
\end{figure}

\begin{figure}[h!]
  \centering
  \includegraphics[width=0.95\textwidth]{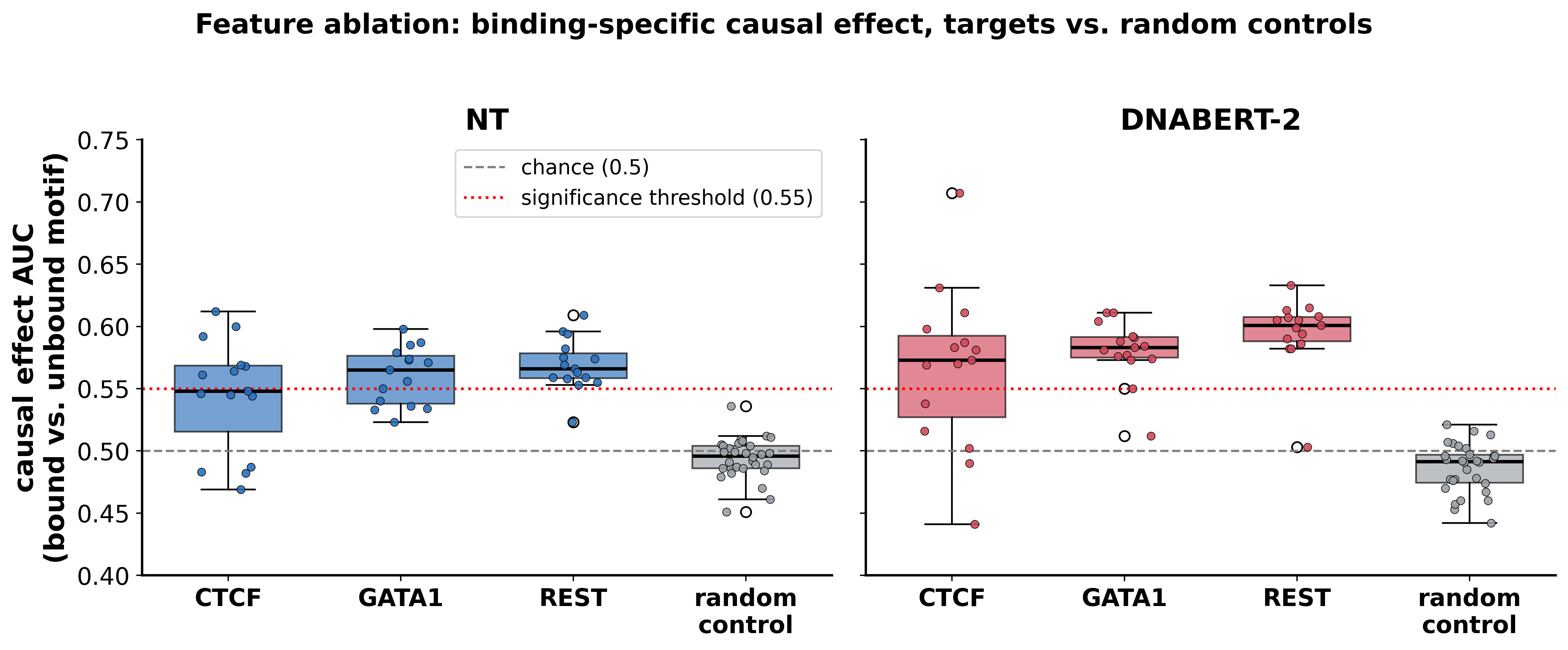}
  \caption{\textbf{Causal effect sizes separate binding features from controls.}
  Distribution of causal-effect $\AUC$ (bound versus unbound motif) for the top
  $15$ binding features of each transcription factor and for pooled random-feature
  controls, in NT (left) and DNABERT-2 (right). Target features lie above the
  significance threshold ($\AUC=0.55$, red dotted) while random controls cluster
  at chance ($\AUC=0.5$, grey dashed). Points are individual features; boxes show
  median and interquartile range.}
  \label{fig:auc}
\end{figure}

\subsection{Negative controls confirm that the framework reports true nulls}

A validation method is only trustworthy if it can return a negative. In both
negative controls, \textsc{scramble} (random windows, random labels) and the
sharper \textsc{gata1-scram} (real GATA1 motif windows, randomized labels), and
in both architectures, the binding-sensitivity test identified \emph{zero}
qualifying features, leaving nothing for the causal stage to test
(Table~\ref{tab:results}, Fig.~\ref{fig:matrix}, bottom rows). The framework thus
manufactures no signal from structured-but-unlabeled data: the positive results
for real transcription factors reflect genuine binding information, not artifacts
of window construction or feature ablation.

\begin{table}[h!]
\centering
\caption{\textbf{Causally validated binding features across transcription factors,
models and controls.} Entries give the number of features (of $15$ tested) with a
significant binding-specific causal effect (Bonferroni $P<1.7\times10^{-3}$,
$\AUC>0.55$); random-feature controls were $0/15$ throughout. Negative controls
(scrambled labels) yielded no binding-sensitive features to test. ``Motif-sel.''
gives the number of motif-selective features identified by the composition-matched test.
Dashes `---' denote not-applicable cells: the causal test is undefined for negative controls (no binding-sensitive features exist to test), and motif-selectivity is undefined for randomly chosen features.}
\label{tab:results}
\small
\begin{tabular}{l c c c c}
\toprule
& \multicolumn{2}{c}{\textbf{NT (6-mer, layer 14)}} & \multicolumn{2}{c}{\textbf{DNABERT-2 (BPE, layer 6)}}\\
\cmidrule(lr){2-3}\cmidrule(lr){4-5}
\textbf{Condition} & Causal / 15 & Motif-sel. & Causal / 15 & Motif-sel.\\
\midrule
CTCF (insulator)            & $7$  & $265$ & $10$ & $88$\\
GATA1 (activator)           & $9$  & $3$   & $13$ & $0$\\
REST (repressor)            & $7$  & $367$ & $14$ & $135$\\
\midrule
\textsc{scramble} (neg.)    & ---  & $0$   & ---  & $0$\\
\textsc{gata1-scram} (neg.) & ---  & $0$   & ---  & $0$\\
\midrule
Random-feature control      & $0$  & ---   & $0$  & ---\\
\bottomrule
\end{tabular}
\end{table}

\section{Discussion}

Three findings have implications beyond our specific models. First, the dominant
obstacle to interpreting genomic models is not extracting features but
\emph{validating} them: position-weight-matrix enrichment, the field's default,
is confounded by GC content and repetitive elements to the point of returning
hundreds of spurious ``TF features'', a failure we both observed empirically and
proved is guaranteed under mild conditions (Proposition~\ref{prop:confound}).
Composition-matched, binding-resolved testing is therefore not optional but
necessary. Second, association and causation diverge in an informative way. The
GATA1 case exhibits a factor whose binding is causally encoded yet poorly captured
by motif correlation, because the model represents occupancy through contextual
sequence beyond the canonical motif; only intervention reveals this, and
Proposition~\ref{prop:sep} explains why magnitude-based attribution would have
missed it. Third, the emergence of binding features in a model trained purely on
reference sequence, with no cell-type labels, indicates that self-supervised
genomic models internalize determinants of cell-type-specific occupancy latent in
sequence context, an observation that invites systematic study.

Our analysis has limitations that also define natural extensions. We focused on
two masked-language-model encoders; autoregressive and state-space genomic
models~\citep{nguyen2024evo,schiff2024caduceus,gu2023mamba} expose different
readouts and merit dedicated treatment. We examined three well-characterized
factors; the framework scales directly to the hundreds of factors with ENCODE
ChIP-seq, enabling a systematic atlas of which regulatory programs genomic models
encode and use. Finally, the same causal machinery could be turned from
validation to steering, editing binding features to test sequence-design
hypotheses in silico.

\section{Conclusion}

We presented a purely computational framework that extracts interpretable features
from genomic language models, validates them against experimental binding data
while controlling for the compositional confounds that plague naive analyses, and
tests by direct intervention whether the model actually uses each feature to
represent binding. The framework recovers reproducible, causally validated
transcription-factor binding features across two architectures with different
tokenizations and across three transcription factors of distinct regulatory
function, while returning clean nulls under two classes of negative control. By
coupling sparse dictionary learning to causal intervention and disciplined
controls, it offers the regulatory-genomics community a reusable standard for
interpretability claims: a feature is real not when it correlates with a motif,
but when the model demonstrably uses it.





\bibliographystyle{unsrtnat}
\bibliography{references}

\end{document}